\documentclass[letterpaper, 10 pt, conference]{ieeeconf}
\IEEEoverridecommandlockouts
\usepackage{graphics,epsfig,rotating,setspace,latexsym,amsmath,epsf,amssymb,bm,theorem}
\usepackage{cite,enumerate,color}
\usepackage{authblk,subcaption}
\usepackage{lipsum}

\newtheorem{theorem}{Theorem}
\newtheorem{corollary}{Corollary}
\newtheorem{lemma}{Lemma}
\newtheorem{proposition}{Proposition}

\newcommand{\mb}{\mathbb}
\newcommand{\bl}{\boldsymbol}

\newcommand{\sgn}{\text{sign}}
\DeclareMathOperator{\arccot}{arccot}

\begin{document}

\title{\LARGE \bf Supermajority Sentiment Detection with External Influence in Large Social Networks}
\author{Tian Tong$^{1}$ and Rohit Negi$^{2}$
\thanks{*This research was partially supported by NSF awards CCF1422193 and CNS1218823.}
\thanks{Electrical and Computer Engineering Department, Carnegie Mellon University, {$^{1}$\tt\small ttong1@andrew.cmu.edu}, {$^{2}$\tt\small negi@ece.cmu.edu} 
}}

\maketitle
\thispagestyle{empty}
\pagestyle{empty}

\begin{abstract}
In a large social network whose members harbor binary sentiments towards an issue, we investigate the asymptotic accuracy of sentiment detection. We model the user sentiments by an Ising Markov random field model and allow the user sentiments to be biased by an external influence. We consider a general supermajority sentiment detection problem and show that the detection accuracy is affected by the network structure, its parameters, as well as the external influence level.  
\end{abstract}

\IEEEpeerreviewmaketitle

\section{Introduction}
\label{sec:introduction}
Social networks have become important sources for politicians, sociologists, and financial analysts alike to detect, predict and shape the sentiments of people. A practical example can be a polling firm that predicts the outcome of an election by detecting sentiments from Tweets, where the underlying network can be Twitter follower/followee relations. Sociology has long been fascinated by the power of social networks, and asked questions about how social networks can influence peoples' sentiments \cite{kouloumpis2011twitter,pang2008opinion}. A variety of empirical work has shown the effect of networks to lead sentiments to the two different types of steady states: consensus or dissent. At the same time, several models have been proposed to theoretically explain why social networks can lead to dramatically different behaviors. 

We model the social network as an Ising Markov random field, with members as vertices, relation between members as edges, and sentiments as random variables taking the values $-1$ or $+1$. We are interested in large social networks, which we model by considering a sequence of graphs whose size grows to infinity. This will allow us  to describe macro-behavior in the asymptotic limit of large social networks. Ising models have been extensively studied in statistical physics \cite{lee1952statistical,baxter2007exactly,friedli_velenik_2017} as simple models that capture the essence of atomic interactions where global properties emerge based on local interactions of atoms with their neighbors. They have been used to explain the phase transition phenomena in materials such as iron, which turns into a permanent magnet at suitably low temperatures (ferromagnetic phase) but is unable to do so at higher temperatures (paramagnetic phase). 

Previously, we had considered the problem of {\em majority} sentiment detection without external influence and showed that the detection error probability of such sentiment demonstrates complex behaviors \cite{tongasymptotic}. Using certain simple examples, we showed the counter-intuitive result that the error probability asymptotically decreases to zero for certain networks, but remains bounded positive for others. In this paper, we generalize that problem in three ways. Firstly, we consider {\em supermajority} sentiment detection as an important problem in certain applications. Supermajority sentiment of members towards an issue is the condition that one  sentiment  predominates among the members, with a level of support greater than the typical threshold of one-half used to define majority sentiment. An example is the requirement for the passage of constitutional amendments in the US, or the decision by companies to develop expensive features if an overwhelming fraction of its users demand it. Secondly, we allow for an external influence that biases the sentiment of the network members in the positive or negative direction. We would like to investigate the interaction of supermajority detection with the strength of the external influence. We will recover the results in \cite{tongasymptotic} as a special case of this analysis. Finally, we present results that allow analysis of the problem in certain graphs where closed form solutions are not available (such as Lattice graphs).

In this paper, we attempt to answer the following questions:
\begin{itemize}
\item
What is the relation between supermajority sentiment detection and the distribution of average member sentiments?
\item
Given a supermajority threshold level $S$, do the average sentiments concentrate near $S$ or far from $S$?
\item
What is the asymptotic supermajority detection error probability for various stylized social networks? 
\end{itemize}

We show that the asymptotic performance of supermajority sentiment detection is closely related to the asymptotic distribution of average member sentiments. In situations where the average sentiments stay away from $S$, the detection is asymptotically accurate, while if the average sentiments stay near $S$, the detection is asymptotically inaccurate. In turn, the 
distribution of average  sentiments depends on the graph structure, the external influence, and in some graphs, also on the strength of the member connections to each other.

The paper is organized as follows. Section \ref{sec:model} introduces the supermajority sentiment detection problem formally. Section \ref{sec:msd} presents the main results of the paper as a sequence of theorems that allow analysis of the detection error probability under various assumptions on the social network model. Section \ref{sec:networks} uses the previous section's results to calculate the detection error probability for various examples of networks, and shows that the error probability exhibits counter-intuitive behavior in many cases. Section \ref{sec:num} presents numerical results and Section \ref{sec:conclusion} concludes the paper.

\section{System Model}
\label{sec:model}
The social network structure is modeled as an undirected graph, as shown in Fig.~\ref{fig_model}.
\begin{figure}[ht]
\vspace*{-0.6cm}
\centerline{\includegraphics[width=0.8\linewidth]{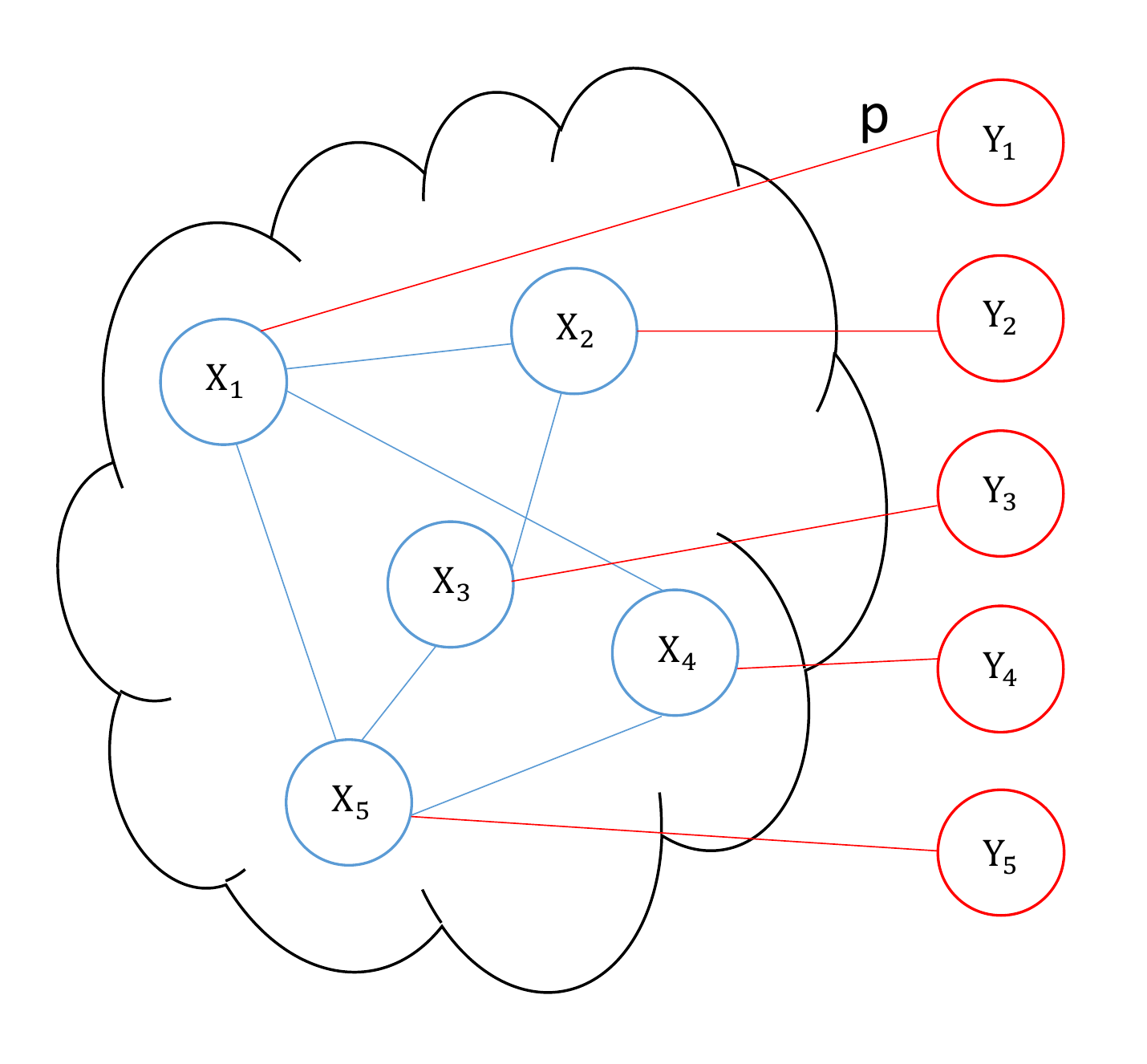}}
\vspace*{-0.5cm}
\caption{Markov random field model of sentiment detection.}
\label{fig_model}
\end{figure}
Let $\bl{X} = (X_1, \dots, X_n)^T \in \{-1, 1\}^n$ denote the vector of binary sentiments of $n$ members, where $1$ and $-1$ denote positive or negative sentiments, respectively. We adopt a homogeneous Ising Markov random field (MRF) prior on $\bl{X}$ as:
\begin{align}
\label{eq:isingPrior}
p({\bl x}) = \frac{\exp(\beta {\bl{x}}^TA_n{\bl{x}} + h {\bl 1}^T{\bl x})}{Z_n(\beta, h)}.
\end{align}
Here ${\bl 1}$ is the vector of all-ones. $A_n$ denotes the symmetric adjacency matrix of the graph, with $[A_n]_{ij}=0/1$ representing absence/presence of an edge between members $i,j$. $\beta > 0$ is the inverse temperature parameter, characterizing the connection strength, i.e., connected members are more probable to share the same sentiment if  $\beta$ is larger. $h$ is the external influence strength, i.e., sentiments are biased to be positive (negative) by an external influence source if $h > 0$ ($h < 0$). $h$ can be used to model, for example, the societal bias of the network users towards the topic of interest (such as liberal societies generally favoring liberal policies). The normalizer $Z_n(\beta, h) = \sum_{\bl{x}\in\{-1,1\}^n}\exp(\beta {\bl x}^T A_n {\bl x} + h {\bl 1}^T{\bl x})$ is called the partition function. As a special case, the
Empty graph (no edges) is equivalent to assuming that sentiments are independent and identically distributed (i.i.d.) random variables, which is assumed in typical polling analysis. The key results in this paper will be asymptotic, for which we assume that in a given network example, increasing $n$ results in a sequence of graphs, with associated adjacency matrices $A_n$.

In the paper, we denote $\overline{X_n} = \frac{1}{n}\sum_{i=1}^n X_i$ as the sample average (equivalently, the average member sentiments), $\mb{E}[\cdot]$ and $\mb{V}[\cdot]$ as the mean and variance, and $\xrightarrow{d}$ as convergence in distribution.

\subsection{Supermajority Sentiment Detection}
Given a supermajority threshold level $S \in (-1, 1)$, the supermajority sentiment is defined as the Bernoulli variable
\begin{align}
m = \sgn(\overline{X_n} - S),
\end{align}
Thus, the supermajority sentiment is considered to be positive if more than $(1+S)/2$ fraction of sentiments are positive. Here we assume that $n(1+S)/2$ is not integer to avoid trivial ambiguity. For example, if $S = 0$, this degenerates to the typical majority sentiment variable (more than one-half the sentiments are positive), while if $S=1/3$, we are considering a supermajority of two-thirds (such as the fraction of votes needed in the US Congress to over-ride a Presidential veto). 

Since sentiments $\bl{X}$ are unknown to observers, we hope to estimate the supermajority sentiment by using noisy observations of $\bl{X}$, called {\em measured sentiments} $\bl{Y}$. $\bl{Y} = (Y_1, \dots, Y_n)^T \in \{-1, 1\}^n$ is modeled as conditionally independent binary measurements of $\bl{X}$ each with cross-over probability $p$, i.e., the output of a binary symmetric channel with input $\bl{X}$. We assume that $p < 1/2$. The error in the measured sentiment could arise due to error in automatic language analysis (as in the case of automated Twitter tweet analysis) or due to intentional prevarication by users in responding to a pollster.

The Maximum Aposteriori (MAP) detector of the supermajority sentiment is difficult to calculate due to the curse of dimensionality, and is even harder to analyze. Therefore, we will use a naive (but reasonable) detector for the supermajority sentiment:
\begin{align}
\label{eq:majorDetector}
\hat{m} = \sgn(\overline{Y_n} - (1-2p)S).
\end{align}
This detector estimates the supermajority sentiment as the supermajority of noisy measurements with level $(1-2p)S$. The detector is called `naive' because it does not use the knowledge of the network, neither the adjacency matrix $A_n$ nor parameters $\beta,h$. The detector is based on the result that $\mb{E}[\overline{Y_n}|{\bf X}] = (1-2p)\overline{X_n}$, so that $\overline{X_n} \gtrless S$
is equivalent to $\mb{E}[\overline{Y_n}|{\bf X}] \gtrless (1-2p)S$. Further, this can be proved to be the MAP detector in the case of Empty graph (i.e., i.i.d. sentiments) in the limit of large $n$. Also, note that in the more common case of majority sentiment detection ($S=0$), this detector reduces to deciding whether the sum of sentiments is positive by simply checking whether the sum of measured sentiments is positive.  In general, it is not the optimal MAP detector. However, it will be sufficient to illustrate the key insights of this paper. 

The detection error probability (equivalently, classification error probability of the two sentiment class problem) is 
\begin{align}
\label{eq:errorProb}
P^{(n)}_e = \mb{P}(m \neq \hat{m}).
\end{align}
We hope to investigate whether the supermajority sentiment detection is asymptotically accurate, i.e., whether $P^{(n)}_e$ becomes arbitrarily small when $n$ is sufficiently large.

\section{Error performance of supermajority sentiment detection}
\label{sec:msd}

In \cite{tongasymptotic}, we showed the surprising result that majority sentiment detection without external influence is not necessarily accurate, even in the limit of large $n$. In fact, in certain graphs, that problem shows subtle behavior, switching from asymptotic accuracy to bounded accuracy as the connection strength between the social network members decreases. Thus, it is not apriori  obvious when supermajority sentiment detection in the presence of external influence, as considered in this paper, will be  accurate.
In this section, we analyze the asymptotic accuracy of supermajority sentiment detection. Since it is difficult to find closed form results on the error performance of such detection, we present a set of theorems that allow us to analyze the performance for different types of graphs. These theorems show that the asymptotic performance of the detection error probability is related to the asymptotic distribution of $\overline{X_n}$ as $n \to \infty$. These results will allow us to analyze the error behavior for various interesting network examples in Section \ref{sec:networks}.

We begin by first obtaining an upper bound on the error probability that applies to any fine $n$, as below.
\begin{theorem} 
\label{thm:Hoeffd} An upper bound on the detection error probability is:
$$P_e^{(n)} \le \mb{E}\left[\exp\left(-C_p(\sqrt{n}(\overline{X_n}-S))^2\right)\right],$$
where $C_p = \frac{1}{2}(1-2p)^2$.
\end{theorem}
\begin{proof}
Let $Z_i = Y_i - (1-2p)X_i$. Since $Z_i$s are conditionally independent given ${\bf X}$, with $\mb{E}[Z_i|{\bf X}] = 0$ and $\max Z_i - \min Z_i \le 2$, Hoeffding's inequality tells us that average $\overline{Z_n}$ satisfies $\mb{P}(\overline{Z_n} > \epsilon \mid {\bf X}) \le \exp\left(-\frac{1}{2}n\epsilon^2\right)$ and $\mb{P}(\overline{Z_n} < -\epsilon \mid {\bf X}) \le \exp\left(-\frac{1}{2}n\epsilon^2\right)$ for any $\epsilon > 0$. By definition (\ref{eq:errorProb}):
\begin{align}
\label{eq:PeFormula}
P_e^{(n)} &= \mb{P}((\overline{X_n} - S)(\overline{Y_n} - (1-2p)S)< 0) \nonumber\\
&= \mb{P}((\overline{X_n} - S)\overline{Z_n} < -(1-2p)(\overline{X_n}-S)^2) \nonumber \\
&= \mb{E}[ \mb{P}((\overline{X_n}-S) \overline{Z_n} <  -(1-2p)(\overline{X_n}-S)^2 \mid {\bf X}) ]\\
& \le \mb{E}[\exp(-C_p(\sqrt{n}(\overline{X_n}-S))^2)]. \nonumber 
\end{align}
\end{proof}

In some applications, we may want to estimate the supermajority sentiment by using only partial observations. For example, in large graphs, we may want to poll only a subset of members to save on cost. Define partial observations $Y_{i;\delta} = \begin{cases} Y_i, w.p.~\delta \\0, w.p.~1-\delta \end{cases}$, which means that the $i$th member's sentiment is measured independently with probability $\delta$. The detector based on partial observations is $\hat{m}_\delta = \sgn(\overline{Y_{n;\delta}}-(1-2p)S)$. An upper bound on its detection error probability is as below.
\begin{corollary}
\label{cor:Hoeff}
$P^{(n)}_{e;\delta} \le \mb{E}\left[\exp\left(-\delta^2 C_p(\sqrt{n}(\overline{X_n}-S))^2\right)\right]$.
\end{corollary}

This corollary shows that the effect on error performance of polling only a fixed fraction of members is no worse than polling all members but increasing each member's measurement error probability from $p$ to $((1-\delta)(1-2p)+p)/((1-\delta)(1-2p)+1)$.

While the previous theorem can be used to bound the error probability for finite $n$, it is difficult to calculate the bound in large graphs. However, in an  sequence of graphs of increasing size, it may be possible to calculate  the asymptotic behavior of $\overline{X_n}$. Therefore, we derive  asymptotic upper and lower bounds on error probability in Theorem~\ref{thm:Q}. This will follow from a conditional central limit theorem given in Lemma~\ref{lemma:CLT}, whose proof can be found in \cite{tongasymptotic}.

\begin{lemma}
\label{lemma:CLT} 
For all ${\bf X}$, $$\sqrt{n}(\overline{Y_n} - (1-2p)\overline{X_n}) \mid {\bf X} \xrightarrow{d} N(0, 4p(1-p)).$$ 
\end{lemma}

\begin{theorem}
\label{thm:Q} The superior and inferior limit of the detection error probability is:
\begin{align*}
\liminf_{n\to\infty}P_e^{(n)} &= \liminf_{n\to\infty}\mb{E} \left[Q\left(D_p \sqrt{n} \lvert\overline{X_n}-S \rvert \right) \right] \\
\limsup_{n\to\infty}P_e^{(n)} &= \limsup_{n\to\infty}\mb{E} \left[Q\left(D_p \sqrt{n} \lvert\overline{X_n}-S \rvert \right)\right],
\end{align*} where $D_p = \frac{1-2p}{\sqrt{4p(1-p)}}$, $Q(\cdot)$ is the tail probability of standard normal distribution: $Q(x) = \frac{1}{\sqrt{2\pi}}\int_x^\infty \exp(-t^2/2)dt$.
\end{theorem}

\begin{proof}
Start with equation (\ref{eq:PeFormula}). Define $\varepsilon_n({\bf X}) = \mb{P}(\sqrt{n}(\overline{X_n}-S)\sqrt{n}\,\overline{Z_n} < -(1-2p)(\sqrt{n}(\overline{X_n}-S))^2 \mid {\bf X}) - Q\left(D_p \sqrt{n} \lvert\overline{X_n}-S \rvert \right)$. From the asymptotic distribution in Lemma~\ref{lemma:CLT}, $\lim\limits_{n\to\infty} \varepsilon_n({\bf X}) = 0, \forall {\bf X}$. So,
\begin{align*}
\liminf_{n\to\infty} & P_e^{(n)} = \liminf_{n\to\infty}\mb{E} [\mb{P}(\sqrt{n}(\overline{X_n}-S) \sqrt{n}\,\overline{Z_n} < \\
& \qquad \qquad \qquad -(1-2p)(\sqrt{n}(\overline{X_n}-S))^2 \mid {\bf X}) ]\\
&= \liminf_{n\to\infty} \mb{E} \left[Q\left(D_p \sqrt{n} \lvert\overline{X_n}-S \rvert \right)+\varepsilon_n({\bf X})\right] \\
&= \liminf_{n\to\infty} \mb{E}\left[Q\left(D_p \sqrt{n} \lvert\overline{X_n}-S \rvert \right)\right].
\end{align*}
Here $\liminf\limits_{n\to\infty} \mb{E} \left[\varepsilon_n({\bf X})\right] = \mb{E}\left[ \liminf\limits_{n\to\infty}\varepsilon_n({\bf X}) \right] = 0$ by Lebesgue's dominated convergence theorem \cite{dudley2002real}, since $\lvert \varepsilon_n({\bf X}) \rvert \le 1$. Use the same technique for `$\limsup$' case.
\end{proof}

The asymptotic error probability can be exactly obtained if the exact asymptotic distribution of $\overline{X_n}$ can be obtained, as shown below. Section \ref{sec:networks} shows the productive use of this result in various  graphs of interest.
\begin{corollary}
\label{cor:lim}
\begin{enumerate}[(a)]
\item If $\sqrt{n}(\overline{X_n}- S) \xrightarrow{d} \Phi$, where $\Phi$ is a distribution, then
 $$\lim_{n\to\infty}P_e^{(n)} = \int_{-\infty}^\infty Q\left(D_p\left\lvert x \right\rvert \right) \Phi(dx).$$ 

\item Specifically, if $\sqrt{n}(\overline{X_n} - S)\xrightarrow{d} N(0, \sigma^2)$, then $$\lim_{n\to\infty}P_e^{(n)}= \frac{1}{\pi} \arccot \left(D_p\sigma\right) > 0.$$
\end{enumerate}
\end{corollary}

The above results on asymptotic error probability are difficult to use if the corresponding expectations are difficult to calculate. On the other hand, a critical question in such error analysis is not necessarily the exact value of the error probability, but instead, whether this probability becomes arbitrarily small for large $n$, or whether it is bounded below even with infinite number of members. Based on Theorems~\ref{thm:Hoeffd} and \ref{thm:Q}, we show in Theorem~\ref{thm:tt} that the question of whether the detection error probability tends to $0$ or not, is exactly determined by whether $\overline{X_n}$ asymptotically stays away from $S$ or not. Intuitively, if the probability of $\overline{X_n}$ being near $S$ decays to $0$, then the error probability tends to zero; if it decays to $0$ exponentially fast, then the error probability tends to zero exponentially fast. On the other hand, if the probability of $\overline{X_n}$ being near $S$ remains asymptotically positive, so does the error probability.

\begin{theorem}
\label{thm:tt}
\begin{enumerate}[(a)]
\item \label{part:tt=0}
If $\ \forall B > 0$, $\lim\limits_{n\to\infty}\mb{P}(\sqrt{n}\lvert \overline{X_n} - S \rvert \le B) = 0$, then $\lim\limits_{n\to\infty}P_e^{(n)}= 0$. 

\item \label{part:tt>0}
If $\exists B > 0$ s.t. $\liminf\limits_{n\to\infty}\mb{P}(\sqrt{n} \lvert \overline{X_n} - S\rvert \le B) > 0$, then  $\liminf\limits_{n\to\infty}P_e^{(n)}> 0$.

\item \label{part:ttExp0}
If $\exists b > 0$, $\alpha_b > 0$, s.t. $\limsup\limits_{n\to\infty} \frac{1}{n}\log P(|\overline{X_n} - S| \le b) \le -\alpha_b$, then $\limsup\limits_{n\to\infty} \frac{1}{n}\log P_e^{(n)} \le - \min\{\alpha_b, C_pb^2\}$.
\end{enumerate}
\end{theorem}
\begin{proof}
For part~(\ref{part:tt=0}), given any $\epsilon > 0$, choose $B$ large enough such that $\exp\left(-C_p B^2\right) \le \epsilon/2$, then choose $n$ large enough such that $\mb{P}(\sqrt{n} \lvert \overline{X_n} - S\rvert \le B) \le \epsilon/2$. Theorem~\ref{thm:Hoeffd} tells that 
\begin{align*}
P_e^{(n)} &\le \mb{E}\left[\exp\left(-C_p(\sqrt{n}(\overline{X_n}-S))^2\right)\right] \\
&\le \mb{P}(\sqrt{n} \lvert \overline{X_n}-S\rvert \le B) + \exp\left(-C_pB^2\right) \le \epsilon.
\end{align*}
Thus $\lim\limits_{n\to\infty}P_e^{(n)}= 0$.

For part~(\ref{part:tt>0}), choose a $B > 0$, satisfying that $\liminf\limits_{n\to\infty}\mb{P}(\lvert \sqrt{n}(\overline{X_n}-S)\rvert \le B) > 0$. Theorem~\ref{thm:Q} tells that
\begin{align*}
\liminf_{n\to\infty}& P_e^{(n)} = \liminf_{n\to\infty}\mb{E} \left[Q\left(D_p \sqrt{n} \lvert\overline{X_n}-S \rvert \right)\right] \\
& \ge \liminf_{n\to\infty}\mb{P}(\sqrt{n} \lvert \overline{X_n}-S\rvert \le B) Q\left(D_p B \right) > 0.
\end{align*}

For Part~(\ref{part:ttExp0}), given any $\epsilon > 0$, choose $n$ large enough such that such that $\mb{P}(\lvert \overline{X_n} - S\rvert \le b) \le \exp(-n(\alpha_b - \epsilon))$. Theorem~\ref{thm:Hoeffd} tells that 
\begin{align*}
P_e^{(n)} &\le \mb{E}\left[\exp\left(-C_p(\sqrt{n}(\overline{X_n}-S))^2\right)\right] \\
&\le \mb{P}(\lvert \overline{X_n}-S\rvert \le b) + \exp\left(-C_pnb^2\right) \\
&\le \exp(-n(\alpha_b-\epsilon)) + \exp\left(-n C_pb^2\right).
\end{align*}
Thus $\limsup\limits_{n\to\infty}\frac{1}{n}\log P_e^{(n)} \le -\min\{\alpha_b-\epsilon, C_pb^2\}$. Finally let $\epsilon \to 0$.
\end{proof}
Theorem \ref{thm:tt} is easier to apply than Theorems~\ref{thm:Hoeffd} and \ref{thm:Q}, because it only needs the knowledge of the distribution of $\overline{X_n}$ near $S$ and avoids the need to calculate an expectation.

Finally, in complex cases such as the 2-dimensional Lattice graphs, the asymptotic distribution of $\overline{X_n}$ is unknown, and so, it may be difficult to apply  Theorem~\ref{thm:tt}. The state of the art in statistical physics and information theory does allow calculation of certain moments of $\overline{X_n}$ in some of these graphs. Therefore, we present Theorem~\ref{thm:tEV} below, where the analysis of error probability depends only on the asymptotic behavior of $\overline{X_n}$ near its mean. 
\begin{theorem}
\label{thm:tEV}
\begin{enumerate}[(a)]
\item \label{part:tEV=0}
If $\lim_{n\to\infty}\mb{E}[\overline{X_n}] = \mu \neq S$ and $\lim\limits_{n\to\infty} \mb{V}[\overline{X_n}] = 0$,  $\lim\limits_{n\to\infty}P_e^{(n)}= 0$. 
\item \label{part:tEV>0}
If $\lim\limits_{n\to\infty}\mb{E}[\overline{X_n}] = S$ and $\limsup\limits_{n\to\infty} n\mb{V}[\overline{X_n}] < \infty$,  $\liminf\limits_{n\to\infty}P_e^{(n)} > 0$.
\item \label{part:tEVExp0}
If $\lim\limits_{n\to\infty}\mb{E}[\overline{X_n}] = \mu \neq S$ and $\exists 0 < b < |\mu - S|, \alpha_b > 0$, such that $\limsup\limits_{n\to\infty}\frac{1}{n}\log \mb{P}(|\overline{X_n} - \mu| \ge b) \le \alpha_b$,  $\limsup\limits_{n\to\infty} \frac{1}{n} \log P_e^{(n)} \le -\min\{\alpha_b, C_p(|\mu-S|-b)^2\}$.
\end{enumerate}
\end{theorem}
\begin{proof}
For part~(\ref{part:tEV=0}), given any $B > 0$, choose $n$ large enough such that $B < \sqrt{n}|\mu - S|$, Chebyshev's inequality tells us that
\begin{align*}
\mb{P}(\sqrt{n}|\overline{X_n} - S| \le B) &\le \mb{P}(|\overline{X_n} - \mu| \ge |\mu - S| - B/\sqrt{n}) \\
&\le \frac{\mb{V}[\overline{X_n}]}{(|\mu - S| - B/\sqrt{n})^2}.
\end{align*}
Since $\lim\limits_{n\to\infty}\mb{V}[X_n] = 0$, $\lim\limits_{n\to\infty}\mb{P}(\sqrt{n}\lvert\overline{X_n} - S\rvert \le B) = 0$, and by Theorem~\ref{thm:tt} part~(\ref{part:tt=0}), $\lim\limits_{n\to\infty}P_e^{(n)}= 0$.

For part~(\ref{part:tEV>0}), Chebyshev's inequality tells us that
\begin{align*}
\mb{P}(\sqrt{n}|\overline{X_n} - S| \le B) &= 1 - \mb{P}(\sqrt{n}|\overline{X_n} - S| > B) \\
&\ge 1 - \frac{n\mb{V}[\overline{X_n}]}{B^2}.
\end{align*}
Since $\limsup\limits_{n\to\infty} n\mb{V}[\overline{X_n}] = V < \infty$, choose $B = \sqrt{2V}$, $\liminf\limits_{n\to\infty}\mb{P}(\sqrt{n}\lvert \overline{X_n} - S\rvert \le B) \ge 1/2 > 0$, and by Theorem~\ref{thm:tt} part~(\ref{part:tt>0}), $\liminf\limits_{n\to\infty}P_e^{(n)} > 0$.   
\end{proof}

For part~(\ref{part:tEVExp0}), since 
$$\mb{P}(|\overline{X_n} - S| \le |\mu - S| -b) \le \mb{P}(|\overline{X_n} - \mu| \ge  b),$$
by Theorem~\ref{thm:tt} part~(\ref{part:ttExp0}), the statement holds.

\section{Network Examples}
\label{sec:networks}

In this section, we demonstrate the calculation of the asymptotic error probability of supermajority sentiment detection for various graphs, i.e., (sequence of) adjacency matrices $A_n$, the connection strength $\beta$ and external influence level $h$. The intent is to show the application of the error analysis results of Section \ref{sec:msd}; accurate analysis using Theorems~\ref{thm:Hoeffd} and \ref{thm:Q} is possible when the graph has a strongly symmetric structure (such as Empty graph, Chain graph, etc.). However, the weaker results of Theorems~\ref{thm:tt} and \ref{thm:tEV} apply in more difficult cases, such as the 2-dimensional Lattice graph. Besides demonstrating the value of the error analysis of the previous section, we believe that it is of independent interest to obtain the asymptotic error probability in these network examples, since they show increasingly complex error  behavior.

For several graphs where the partition function can be calculated, the asymptotic distribution of $\overline{X_n}$ can be analyzed thoroughly. A major role in this analysis is played by the \textit{free entropy density}, defined as
\begin{align}
\label{eq:free_entropy_n}
\psi_n(\beta, h) &= \frac{1}{n} \log Z_n(\beta, h).
\end{align}
The mean and variance of $\overline{X_n}$ are the derivatives of the free entropy density with respect to $h$:
\begin{align}
\label{eq:E_n}
\mb{E}[\overline{X_n}] &= \frac{\partial}{\partial h}\psi_n(\beta, h) \\
\label{eq:V_n}
\mb{V}[\overline{X_n}] &= \frac{\partial^2}{n \partial h^2}\psi_n(\beta, h).
\end{align}
We are interested in the asymptotic property of free entropy density. When the limit as $n \to \infty$ exists, define it as:
\begin{align}
\label{eq:free_entropy}
\psi(\beta,h) &= \lim_{n\to\infty} \frac{1}{n} \log Z_n(\beta, h).
\end{align}
Under some regularity conditions where  the limits and the derivatives in (\ref{eq:E_n}) and (\ref{eq:V_n}) can be exchanged, the limit mean and variance can be obtained as
\begin{align}
\label{eq:EVdifferentiation}
\mu_{\beta,h} & = \lim_{n\to\infty}\mb{E}[\overline{X_n}]  = \frac{\partial}{\partial h}\psi(\beta, h) \\
\sigma^2_{\beta,h} & = \lim_{n\to\infty} n\mb{V}[\overline{X_n}] = \frac{\partial^2}{\partial h^2}\psi(\beta, h).
\end{align} 

Furthermore, two kinds of limit theorems, central limit theorem and large deviation theorem, can be established based on the free entropy density function. Firstly, when the limit mean and variance exist,  \cite{bryc1993remark} shows that under some regularity conditions, there exists a central limit theorem:
\begin{align}
\sqrt{n}(\overline{X_n} - \mu_{\beta,h}) \xrightarrow{d} N(0, \sigma^2_{\beta,h}).
\end{align}
Secondly, since $\psi(\beta, h)$ is convex in $h$, we know that its left and right derivative with respect to $h$ always exist. Denote them as $\mu^+_{\beta,h} = \frac{\partial \psi(\beta,h)}{\partial h^+}$ and $\mu^-_{\beta,h} = \frac{\partial \psi(\beta,h)}{\partial h^-}$, respectively. The interval $[\mu^-_{\beta,h}, \mu^+_{\beta,h}]$ is called the phase transition interval. A large deviation theorem \cite{touchette2009large} tells us that the probability of $\overline{X_n}$ falling outside the phase transition interval is only exponentially small. Specifically,  for any $b > 0$,
\begin{align}
\label{eq:deviation}
\limsup_{n\to\infty}&\frac{1}{n}\log \mb{P}(\overline{X_n} - \mu^+_{\beta,h} \ge b) < 0 \\
\limsup_{n\to\infty}&\frac{1}{n}\log \mb{P}(\overline{X_n} - \mu^-_{\beta,h} \le -b) < 0.
\end{align}

In the case that $\psi(\beta,h)$ is differentiable with respect to $h$, the phase transition interval $[\mu^-_{\beta,h}, \mu^+_{\beta,h}]$ shrinks to a point, i.e., the mean $\mu_{\beta,h}$. Then, the probability that $\overline{X_n}$ deviates from its mean is exponentially small. Specifically, for any $b > 0$, 
\begin{align}
\limsup_{n\to\infty}\frac{1}{n}\log \mb{P}(|\overline{X_n} - \mu_{\beta,h}| \ge b) < 0.
\end{align}
From convex analysis, we know that the $h$ where $\psi(\beta,h)$ is not differentiable is at most countable, and in fact, is finite in typical graphs. These points of non-differentiability must be treated separately.

With these results, we can now analyze the supermajority detection error performance in specific graphs. Table \ref{tab:netexamples1} lists the free entropy density for the Empty graph (i.e., i.i.d. sentiments) as well as the Star, Chain, Ring and Wheel graphs. In the Empty, Chain and Ring graphs, the free entropy density is twice differentiable for all $h$, while in the Star and Wheel graphs, it is
twice differentiable for all $h \ne 0$. For these cases, the table lists
the asymptotic mean $\mu_{\beta,h}$ and asymptotic variance $\sigma^2_{\beta,h}$. In the Star and Wheel graphs, the case $h=0$
must be  analyzed carefully, by considering the asymptotic modes $\pm \mu^+_{\beta,0}$.

\begin{table*}[t]
\centering
\vspace*{0.2cm}
\scalebox{0.95}{
    \begin{tabular}{ |c|c|c| c|c| }
     \hline
     Graph & $\psi(\beta, h)$ & $\mu_{\beta,h}, \ h \ne 0$ &
     $\mu^+_{\beta,0}$ &  $\sigma^2_{\beta,h}$   \\ 
     \hline
      \hline
      Empty & $\log(2\cosh(h))$ & $\tanh(h)$ & 0 & $\frac{1}{\cosh^2(h)}$ \\ 
      \hline
      Star & $\log(2\cosh(\beta + |h|))$ & $\tanh(\beta + |h|)\sgn(h)$ & $\tanh(\beta)$ &  $\frac{1}{\cosh^2(\beta + |h|)}$ \\ 
     \hline
      Chain/Ring & \shortstack[l]{$\beta + \log\left(\cosh(h) + \right.$ \\ $\left. \sqrt{\sinh^2(h) + \exp(-4\beta)}\right)$} & $\frac{\sinh(h)}{\sqrt{\sinh^2(h) + \exp(-4\beta)}}$ & 0 & $\frac{\exp(-4\beta)\cosh(h)}{(\sinh^2(h) + \exp(-4\beta))^{3/2}}$ \\ 
     \hline
      Wheel & \shortstack[l]{$\beta + \log\left(\cosh(
    \beta + |h|) + \right.$ \\ $\left.   \sqrt{\sinh^2(\beta + |h|) + \exp(-4\beta)}\right)$} & $\frac{\sinh(\beta + |h|)\sgn(h)}{\sqrt{\sinh^2(\beta + |h|) + \exp(-4\beta)}}$ & $\frac{\sinh(\beta)}{\sqrt{\sinh^2(\beta) + \exp(-4\beta)}}$ & $\frac{\exp(-4\beta)\cosh(\beta + |h|)}{(\sinh^2(\beta + |h|) + \exp(-4\beta))^{3/2}}$\\ 
      \hline
      \end{tabular}
  }
  \caption{Free entropy density $\psi(\beta, h)$, asymptotic mean $\mu_{\beta,h}$ (for $h \ne 0$ case), asymptotic positive mode  $\mu^+_{\beta,0}$, and asymptotic variance $\sigma^2_{\beta,h}$ in certain graphs. In the Empty, Chain and Ring graphs,  $\mu^+_{\beta,0}$
  is mean value $\mu_{\beta,0}$. In the Star and Wheel graphs, for $h=0$, $\sigma^2_{\beta,0}$
  denotes the asymptotic variance conditioned on the center vertex being positive.}
\label{tab:netexamples1}
\end{table*}

\subsection{Empty, Chain and Ring graphs}
These graphs can represent a social community with sparse ties to each other. The error probability behavior in these graphs is relatively simple.
In these graphs, Table \ref{tab:netexamples1} lists the asymptotic
mean and variance. Further, in the Empty graph, since $X_i$s are i.i.d., Hoeffding's inequality tells us that $\mb{P}(|\overline{X_n} - \mu_{\beta,h}| \ge b) \le 2 \exp(-nb^2/2)$. In the Chain graph, Hoeffding's inequality for Markov chain \cite{glynn2002hoeffding} tells that $\mb{P}(|\overline{X_n} - \mu_{\beta,h}| \ge b) \le 2 \exp(-n\alpha_b)$ for some $\alpha_b > 0$.
In the Ring graph, the same result as the Chain graph is obtained by
conditioning on one of the vertices.
Therefore by Theorem~\ref{thm:tEV} and Corollary \ref{cor:lim} we conclude that:
\begin{proposition} 
In the Empty, Chain and Ring graphs,
\begin{enumerate}[(a)] 
\item 
If $S \neq \mu_{\beta,h}$, $\lim\limits_{n\to\infty}P_e^{(n)}= 0$. Furthermore $\limsup\limits_{n\to\infty} \allowbreak \frac{1}{n}\log P_e^{(n)} < 0$. 
\item 
If $S = \mu_{\beta,h}$, $\lim\limits_{n\to\infty}P_e^{(n)} = \frac{1}{\pi}\arccot(D_p \sigma_{\beta,h}) > 0$.
\end{enumerate}
\end{proposition} 
In particular, notice that for majority sentiment detection (i.e., $S=0$), which is
the most common application of vote polling, in the absence of an external
influencing field (so that $\mu_{\beta,0}=0$), the error probability is bounded away from zero. Thus, even if infinite users are polled, we cannot
always predict an election's result!

\subsection{Star and Wheel graphs}
The error probability behavior in these graphs is relatively simple, but
exhibits the strong influence of the center vertex.
The Star graph is composed of an Empty graph on ${\bf X}_{\setminus c}$ along with a center vertex $X_c$ connected to all the other vertices. 
The Wheel graph is composed of a Ring graph on ${\bf X}_{\setminus c}$ along with a center vertex $X_c$ connected to all the other vertices. The center vertex may represent, for example, a celebrity member of the social network. Table \ref{tab:netexamples1} lists the free entropy density in these cases
and also the asymptotic mean and variance when $h \ne 0$.

Now consider the case $h = 0$. In the Star graph, by symmetry we know that $\overline{X_n}$ is distributed as in an Empty graph with parameters $(\beta, \beta)$ or with parameters $(\beta, -\beta)$ each with one-half probability.
The former parameters result in the positive mode  $\mu^+_{\beta,0}$
shown in Table \ref{tab:netexamples1} (which is the mean
of the corresponding Empty graph), while the latter parameters
result in the negative mode $-\mu^+_{\beta,0}$. 

Similarly, for $h = 0$, in the Wheel graph, by symmetry we know that $\overline{X_n}$ is distributed as in a Ring graph with parameters $(\beta, \beta)$ or with parameters $(\beta, -\beta)$, each with one-half probability, resulting
in the modes $\pm \mu^+_{\beta,0}$
shown in Table \ref{tab:netexamples1} (which are the means of the corresponding Ring graphs). 

Thus, when $h =0$, Hoeffding's inequality tells us that in the Star and Wheel graphs, we have
$$\mb{P}(|\overline{X_n} - \mu^+_{\beta,0}| \le b) = \mb{P}(|\overline{X_n} + \mu^+_{\beta,0}| \le b) \ge \frac{1}{2} - \exp(-n\alpha_b)$$
for some $\alpha_b > 0$.

Therefore, all cases are covered by using Theorem~\ref{thm:tEV} and Corollary~\ref{cor:lim} to conclude that:
\begin{proposition}
In the Star and Wheel graphs,
\begin{itemize}
\item When $h \neq 0$:
	\begin{enumerate}[(a)]
	\item 
	If $S \neq \mu_{\beta,h}$, $\lim\limits_{n\to\infty}P_e^{(n)}= 0$. Furthermore $\limsup\limits_{n\to\infty} \allowbreak \frac{1}{n}\log P_e^{(n)} < 0$. 
	\item 
	If $S = \mu_{\beta,h}$, $\lim\limits_{n\to\infty} P_e^{(n)} = \frac{1}{\pi}\arccot(D_p \sigma_{\beta,h}) > 0$.
	\end{enumerate}
\item When $h = 0$:
	\begin{enumerate}[(a)]
	\item 
	If $S \neq \pm \mu^+_{\beta,0}$, $\lim\limits_{n\to\infty}P_e^{(n)}= 0$. Furthermore $\limsup\limits_{n\to\infty} \allowbreak \frac{1}{n}\log P_e^{(n)} < 0$. 
	\item 
	If $S = \mu^+_{\beta,0}$ or $S = -\mu^+_{\beta,0}$, $\lim\limits_{n\to\infty}P_e^{(n)} = \frac{1}{2\pi}\arccot\left(D_p \sigma_{\beta,0} \right) > 0$.
	\end{enumerate}
\end{itemize}
\end{proposition} 
Note that in the $h=0$ case, the critical value for the accuracy of supermajority detection in these graphs
shifts to $\pm \mu^+_{\beta,0} \ne 0$, unlike the graphs in the previous proposition
where the critical value was at $\mu^+_{\beta,0} = 0$. This is due to the strong influence
of the center vertex which is connected to all the other vertices.

\subsection{Complete graph}
The Complete graph, representing a close-knit social community, is possibly the simplest graph that demonstrates
a phase transition behavior, where the strength of the connection
$\beta$ strongly affects the accuracy of supermajority sentiment detection.
In a Complete graph, the corresponding Curie-Weiss \cite{kochmanski2013curie} prior is defined slightly differently, in that the strength is weakened to $\beta/n$, to ensure that the total strength from all neighbors of a vertex remains constant, i.e., does not grow with  number of neighbors $n-1$. With this standard modification, the prior on the sentiments is
\begin{align}
\label{eq:CWprior}
p({\bf x}) = \frac{\exp\left(\frac{\beta}{2n}({\bf 1}^T{\bf x})^2 + h {\bf 1}^T {\bf x}\right)}{Z_n(\beta, h)}.
\end{align}

The free entropy density is 
\begin{align}
\label{eq:CWfree_entropy}
\psi(\beta,h) = \max_{\mu \in [-1,1]} h\mu + \frac{1}{2}\beta \mu^2 + H\left(\frac{1+\mu}{2}\right),
\end{align}
where $H(x) = -x\log x - (1-x)\log(1-x)$ denotes the binary entropy function. If either $h \neq 0$ or $\beta \le 1$  the maximizing value $\mu$ in
(\ref{eq:CWfree_entropy}) is unique, denoted as $\mu_{\beta,h}$
(which can be obtained by numerical optimization). In particular,  $\mu_{\beta,0}=0$ in this case. 

When $h = 0$ and $\beta > 1$, there are two maximizing values which are symmetric around $0$, with the positive value being denoted as $\mu^+_{\beta,0}\ne 0$. From results in \cite{dembo2010gibbs}, if either $h \neq 0$ or $\beta \le 1$, there exists $\alpha_b > 0$, such that for all $n$ large enough, 
$$\mb{P}(\lvert \overline{X_n} - \mu_{\beta,h} \rvert \le b) \ge 1 - \exp(-n\alpha_b).$$
In contrast, if $\beta > 1$ and $h = 0$, there exists $\alpha_b > 0$, such that for all $n$ large enough, 
$$\mb{P}(\lvert \overline{X_n} - \mu^+_{\beta,0} \rvert \le b) = \mb{P}(\lvert \overline{X_n} + \mu^+_{\beta,0}| \le b)\ge \frac{1}{2} - \exp(-n\alpha_b).$$

In all cases, we can calculate that $$\frac{\partial^2 }{\partial h^2}\psi(\beta, h) = \sigma^2_{\beta,h} = \frac{1 - \mu^2_{\beta,h}}{1 - \beta + \beta \mu^2_{\beta,h}}.$$
Therefore by Theorem~\ref{thm:tEV} and Corollary~\ref{cor:lim}  we conclude that:
\begin{proposition}
In a Complete graph with Curie-Weiss prior,
\begin{itemize}
\item When $h \neq 0$ or $\beta \le 1$:
	\begin{enumerate}[(a)]
	\item 
	If $S \neq \mu_{\beta,h}$, $\lim\limits_{n\to\infty}P_e^{(n)}= 0$. Furthermore $\limsup\limits_{n\to\infty} \allowbreak \frac{1}{n}\log P_e^{(n)} < 0$. 
	\item 
	If $S = \mu_{\beta,h}$, $\lim\limits_{n\to\infty}P_e^{(n)} = \frac{1}{\pi}\arccot(D_p\sigma_{\beta,h}) > 0$.
	\end{enumerate}
\item When $h = 0$ and $\beta > 1$:
	\begin{enumerate}[(a)]
	\item 
	If $S \neq \pm\mu^+_{\beta,0}$, $\lim\limits_{n\to\infty}P_e^{(n)}= 0$. Furthermore $\limsup\limits_{n\to\infty} \allowbreak \frac{1}{n}\log P_e^{(n)} < 0$. 
	\item 
	If $S = \mu^+_{\beta,0}$ or $S = -\mu^+_{\beta,0}$, $\lim\limits_{n\to\infty}P_e^{(n)} = \frac{1}{2\pi}\arccot(D_p\sigma_{\beta,0}) > 0$.
	\end{enumerate}
\end{itemize}
\end{proposition} 
Unlike the graphs in the previous propositions, notice that the Complete graph demonstrates a phase transition-like behavior in error performance when $S=h=0$; the error performance switches from asymptotically accurate to inaccurate when
the connection  strength $\beta$ drops below the critical value $\beta_c=1$.

\subsection{Lattice graph} 
The $2$-dimensional Lattice graph is historically interesting because
it is the first graph possessing finite degree that was demonstrated to exhibit
a phase transition behavior, thus validating its use as a model to explain ferromagnetism.
In the Lattice graph, \cite{friedli_velenik_2017} provides a detailed analysis of the  free entropy density. When $h \neq 0$ or $\beta < \beta_c = \frac{1}{2}\log(1+\sqrt{2})$, the free entropy density is  differentiable with respect to $h$, so that
the asymptotic mean $\mu_{\beta,h}$ can be obtained by differentiation 
(\ref{eq:EVdifferentiation}). In particular,
$\mu_{\beta,0}=0$ when $\beta < \beta_c$. In contrast, when $h = 0$ and $\beta > \beta_c$, the free entropy density is not differentiable, while the asymptotic properties of $\overline{X_n}$ in the phase transition interval cannot be determined only through the free entropy density. However, statistical physicists typically analyze the Lattice under the positive boundary condition (where the boundary vertices $i$ are all clamped to $X_i=+1$). In that case, a celebrated result in statistical physics shows that
$$\lim_{n\to\infty}\mb{E}[\overline{X_n}] = \mu^+_{\beta,0} = \left. \frac{\partial }{\partial h^+} \psi(\beta, h) \right|_{h=0} > 0.$$
Thus, when $\beta > \beta_c$, the asymptotic mean is positive even without any external influence.
Furthermore, when $\beta$ is sufficiently larger than $\beta_c$, there is a covariance decay result in the Lattice graph stating that $\limsup_{n\to\infty} n \mb{V}[\overline{X_n}] < \infty$. Thus, by Theorem~\ref{thm:tEV}, we have the following results:

\begin{proposition}
In a  2-dimensional Lattice graph, 
\begin{itemize}
\item When $h \neq 0$ or $\beta < \beta_c = \frac{1}{2}\log(1+\sqrt{2})$: 
	\begin{enumerate}[(a)]
	\item 
	If $S \neq \mu_{\beta,h}$, $\lim\limits_{n\to\infty}P_e^{(n)}= 0$. Furthermore $\limsup\limits_{n\to\infty} \allowbreak \frac{1}{n}\log P_e^{(n)} < 0$. 
	\item 
	If $S = \mu_{\beta,h}$, $\liminf\limits_{n\to\infty}P_e^{(n)} > 0$.
	\end{enumerate}
\item When $h = 0$, $\beta$ is sufficiently larger than $\beta_c$, and under the positive boundary condition:
	\begin{enumerate}[(a)]
	\item 
	If $S \neq \mu^+_{\beta,0}>0$, $\lim\limits_{n\to\infty}P_e^{(n)}= 0$. 
	\item 
	If $S = \mu^+_{\beta,0}$, $\liminf\limits_{n\to\infty}P_e^{(n)} > 0$.
	\end{enumerate}
\end{itemize}
 (The same result holds for the negative boundary condition at supermajority threshold level
 $-\mu^+_{\beta,0}$.)
\end{proposition}
Notice that the (finite degree) Lattice graph also demonstrates a subtle error performance behavior, similar to the (infinite degree) Complete graph. When $S=h=0$,
the detection is asymptotically accurate when $\beta$ is sufficiently large, but is inaccurate when
$\beta < \beta_c$. The social network underlying the sentiments matters significantly in this case!

For all the graph examples discussed above, we can make the following general observations.
\begin{itemize}
\item When $S =0 $ and $h \neq 0$, the detection error probability always decays to $0$, i.e., the majority sentiment detection in those networks  is always asymptotically accurate if
there is a non-zero external influence.
\item When $h =0 $ and $S \neq 0$, the detection error probability always decays to $0$, except perhaps for a few (unlucky) choices of $S$, i.e., strict supermajority sentiment detection in those networks  is nearly always asymptotically accurate if
there is no external influence.
\end{itemize}

\section{Numerical Results}
\label{sec:num}
In this section, we provide numerical results on the asymptotic behavior of average member sentiments  $\overline{X_n}$ and on supermajority sentiment detection performance. 

\begin{figure}[ht]
  \centering
  \vspace*{-0.3cm}
  \includegraphics[width=0.85\linewidth]{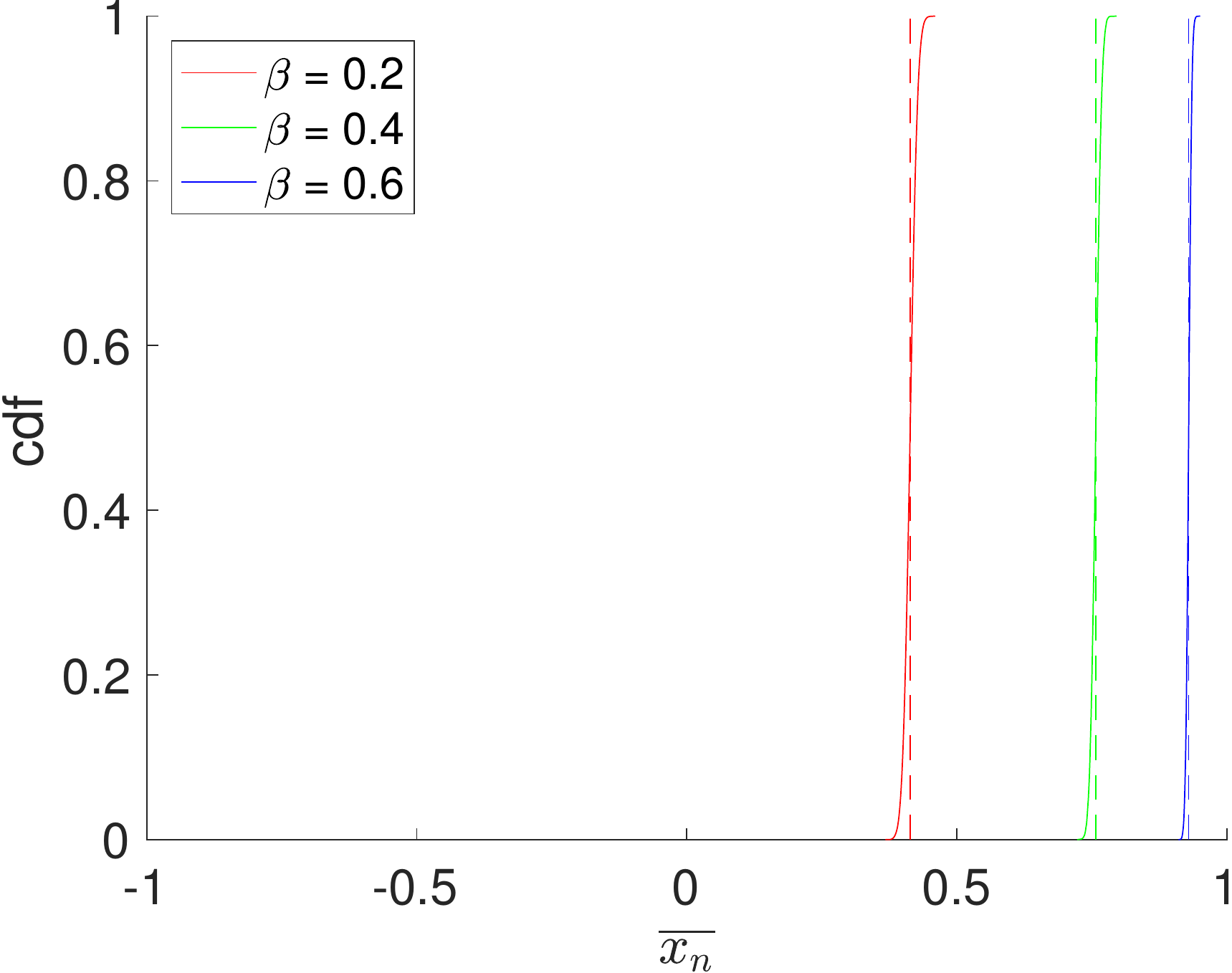}
  \vspace*{-0.2cm}    
  \caption{Wheel graph - part (a): cdf of $\overline{X_n}$ for $n = 10001$, $\beta = \{0.2, 0.4, 0.6\}$, $h = 0.1$, $S = \mu_{\beta,h}$, $p=0.3$. Dashed lines denote $\mu_{\beta,h}$. }
  \label{fig:cdf_wheel_a}
  \vspace*{-0.7cm}
\end{figure}
\begin{figure}[ht]
  \centering
  \includegraphics[width=0.85\linewidth]{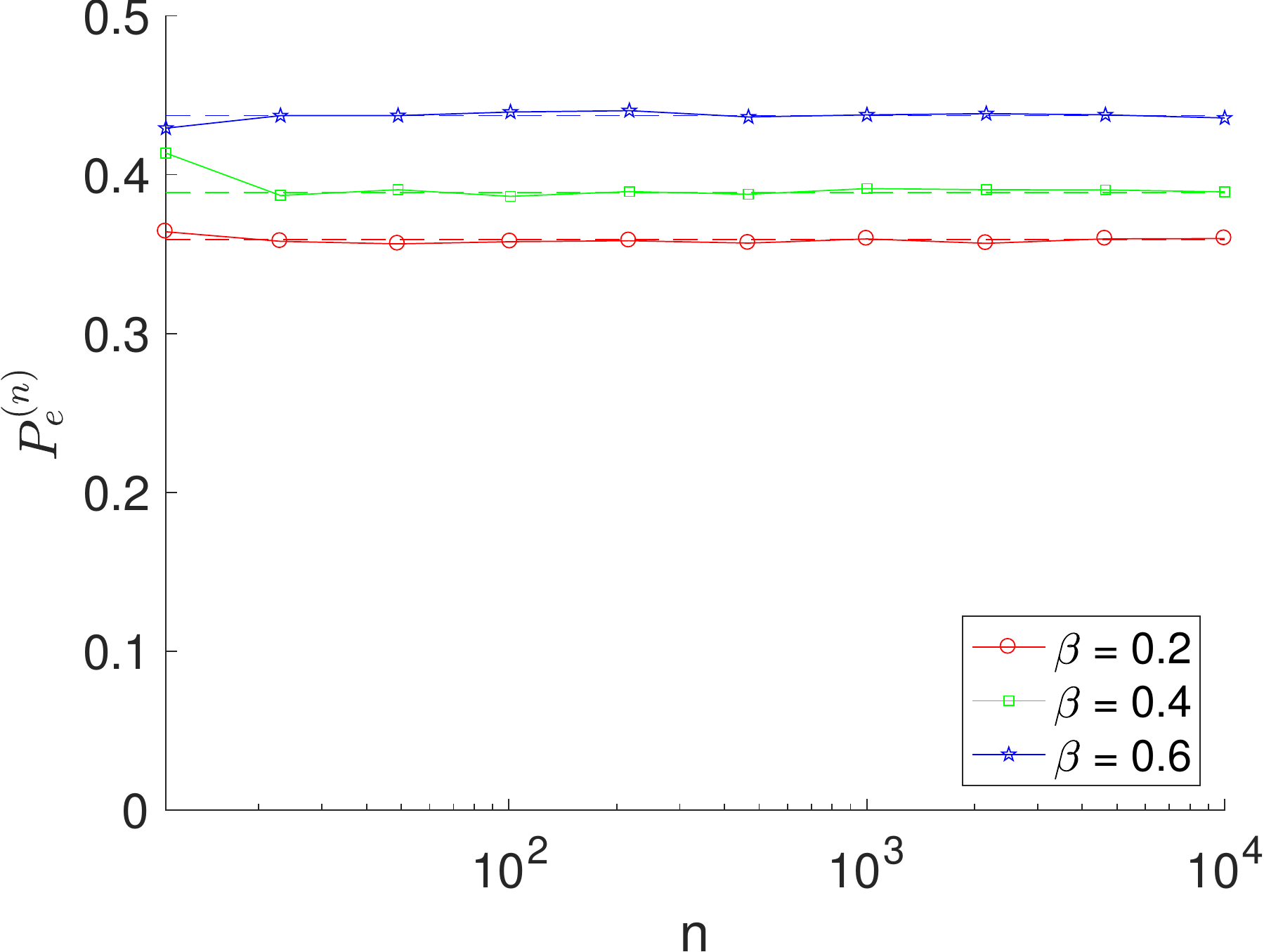}  
  \vspace*{-0.3cm}
  \caption{Wheel graph  - part (a): $P_e^{(n)}$ versus $n$ for $\beta = \{0.2, 0.4, 0.6$\}, $h = 0.1$, $S = \mu_{\beta,h}$, $p=0.3$. Dashed lines denote $\frac{1}{\pi}\arccot(D_p\sigma_{\beta,h})$.}
  \label{fig:error_n_wheel_a}
  \vspace*{-0.3cm}
\end{figure}
First, we consider the Wheel graph. In part (a), we set $h=0.1$ and $S = \mu_{\beta, h} \ne 0$. The cumulative density function (cdf) is shown in Fig.~\ref{fig:cdf_wheel_a}. The detection error probability versus number of members $n$ is shown in Fig.~\ref{fig:error_n_wheel_a}. We can observe that, the distribution of $\overline{X_n}$ concentrates around its mean $\mu_{\beta,h}$. Consequently, the detection error probability does not decay to $0$, but instead converges to $\frac{1}{\pi}\arccot(D_p\sigma_{\beta,h})>0$, as predicted by the results in Section \ref{sec:networks}. 

For the Wheel graph, in part (b), we set $h=S = 0$. The cdf of $\overline{X_n}$ is shown in Fig.~\ref{fig:cdf_wheel_b}. The detection error probability versus $n$ is shown in Fig.~\ref{fig:error_n_wheel_b}. We can observe that the distribution of $\overline{X_n}$ concentrates around the two modes $\pm\mu^+_{\beta,h} \ne 0$, due to the strong influence of the center vertex. However, since the supermajority threshold level $S$ is not equal to one of these modes, the detection error probability reduces to $0$ exponentially fast. 
\begin{figure}[ht]
  \centering
  \includegraphics[width=0.85\linewidth]{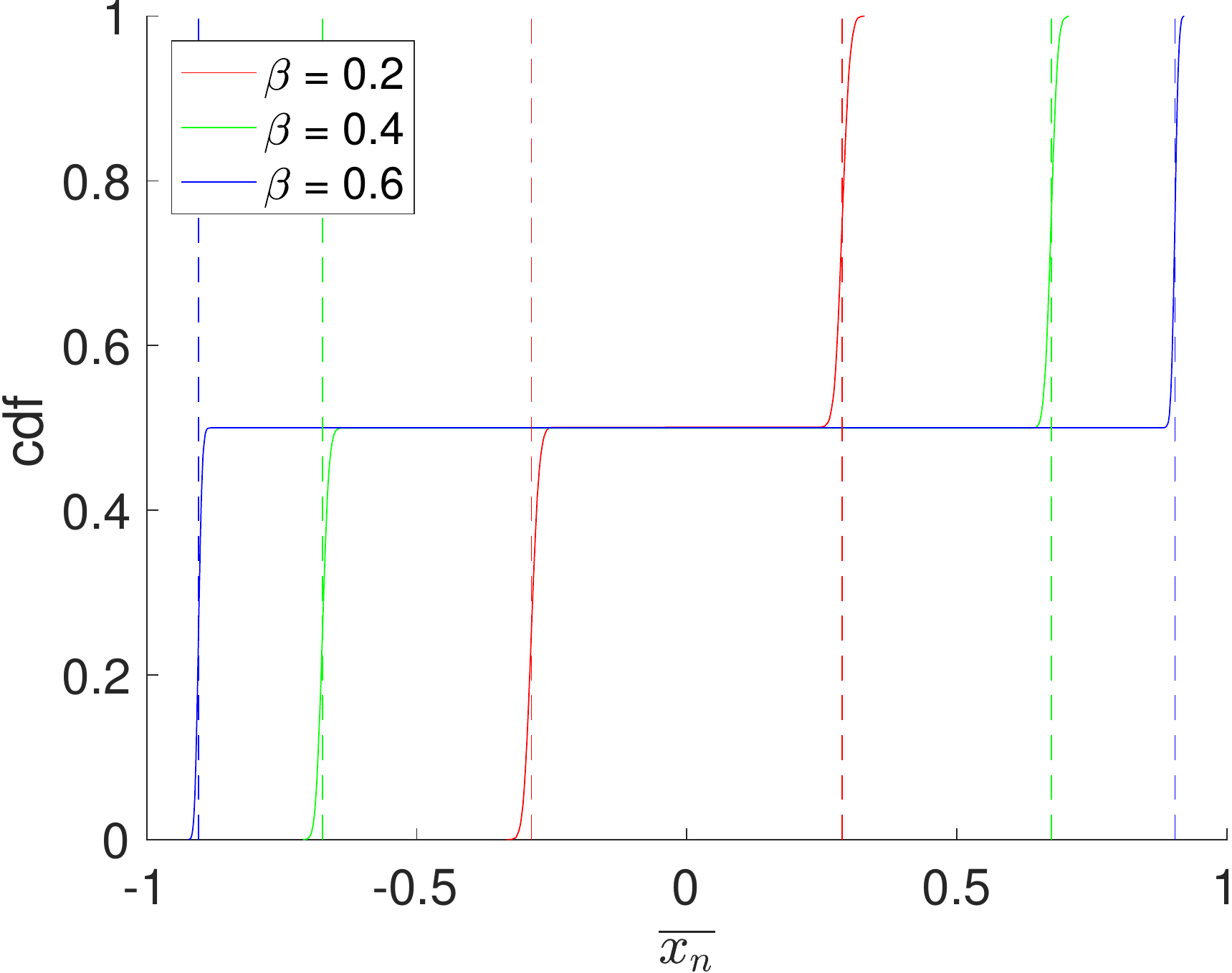}  
  \vspace*{-0.2cm}
  \caption{Wheel graph  - part (b): cdf of $\overline{X_n}$ for $n = 10001$, $\beta = \{0.2, 0.4, 0.6\}$, $h = 0$, $S = 0$, $p=0.3$. Dashed lines denote $\pm\mu^+_{\beta,0}$. }
  \label{fig:cdf_wheel_b}
  \vspace*{-0.3cm}
\end{figure}
\begin{figure}[ht]
  \centering
  \includegraphics[width=0.85\linewidth]{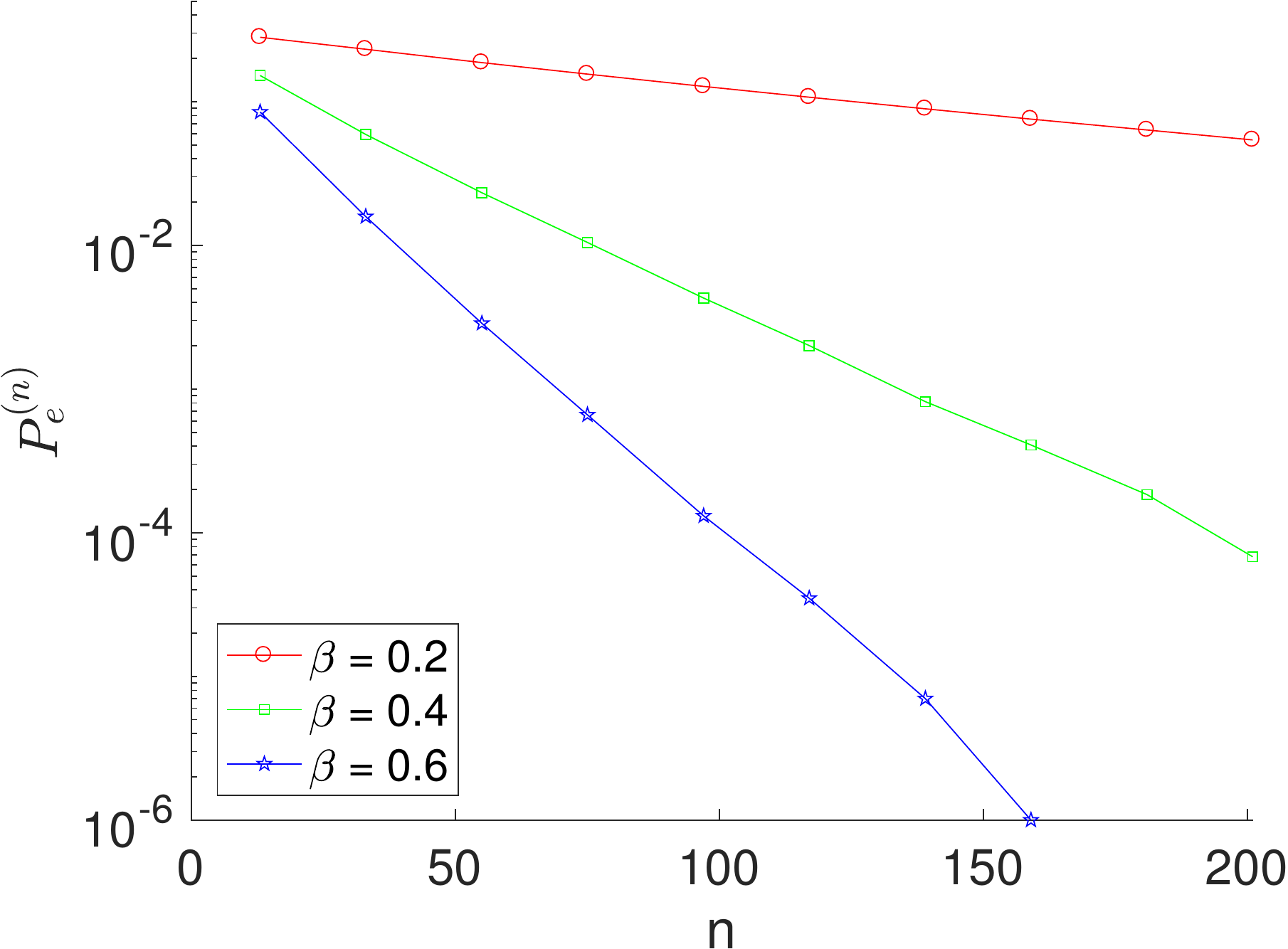} 
  \vspace*{-0.3cm}
  \caption{Wheel graph  - part (b): $P_e^{(n)}$ versus $n$ for $\beta = \{0.2, 0.4, 0.6\}$, $h = 0$, $S = 0$, $p=0.3$. }
  \label{fig:error_n_wheel_b}
  \vspace*{-0.6cm}
\end{figure}

Finally, we consider the Lattice graph for $h=S=0$. The cdf of $\overline{X_n}$ is shown in Fig.~\ref{fig:cdf_lattice}. The detection error probability versus connection strength $\beta$ is shown in Fig.~\ref{fig:error_beta_lattice}. We can observe that, when $\beta < \frac{1}{2}\log(1+\sqrt{2})$, the distribution of $\overline{X_n}$ concentrates around zero, so that the error probability is large (and is predicted in Section \ref{sec:networks} to remain nonzero, even for large $n$.) When $\beta > \frac{1}{2}\log(1+\sqrt{2})$, the distribution concentrates around the two modes which are symmetric around zero. Consequently, the error probability in this regime is small (and is predicted to reduce to zero as  $n\rightarrow \infty$.) 

\begin{figure}[ht]
  \centering
  \includegraphics[width=0.85\linewidth]{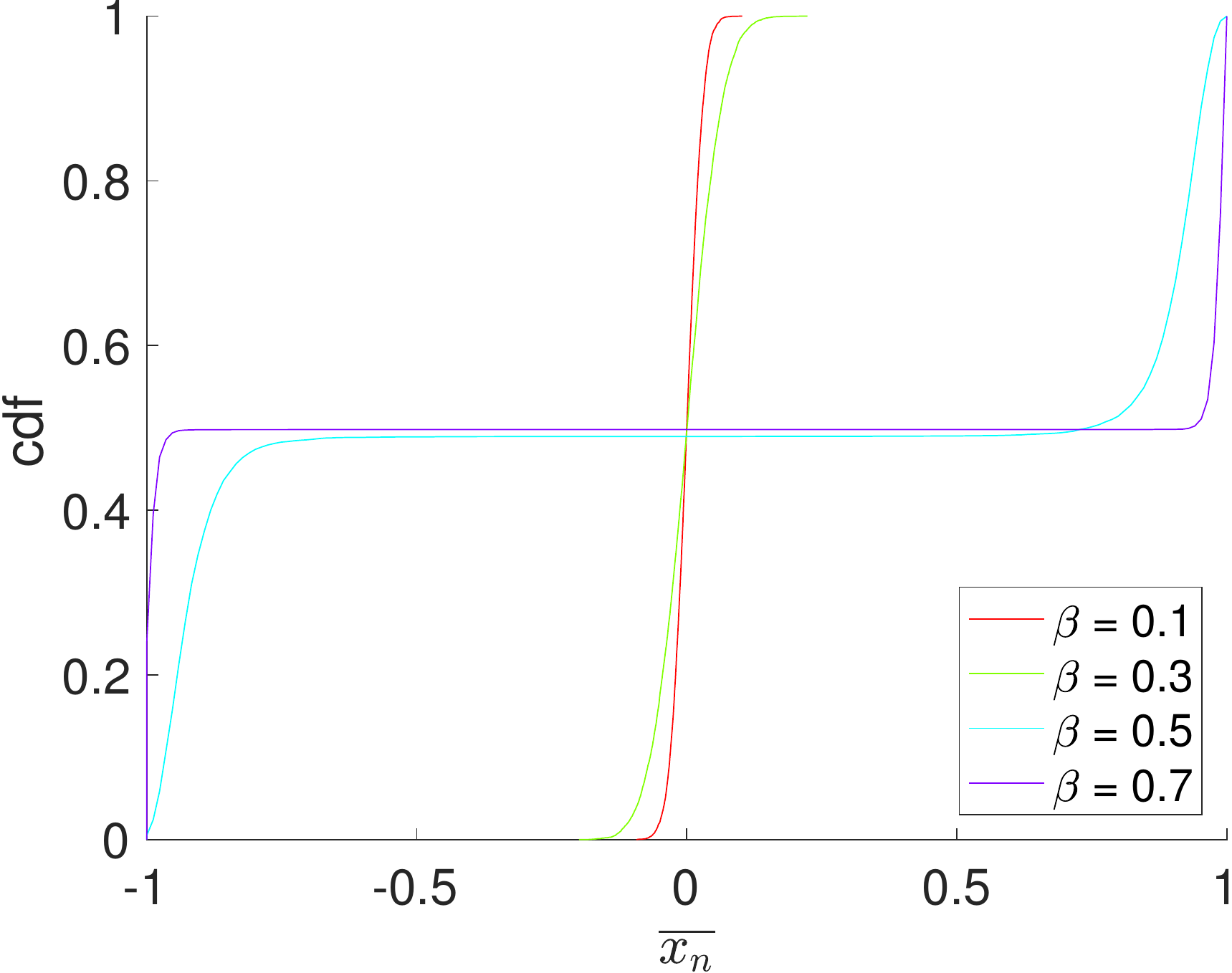}  
  \vspace*{-0.2cm}
  \caption{Lattice graph: cdf of $\overline{X_n}$ for $n = 2601$, $\beta = \{0.1,0.3,0.5,0.7\}$, $h=0$, $S=0$, $p=0.3$.}
  \label{fig:cdf_lattice}
  \vspace*{-0.3cm}
\end{figure}

\begin{figure}[ht]
  \centering
  \includegraphics[width=0.85\linewidth]{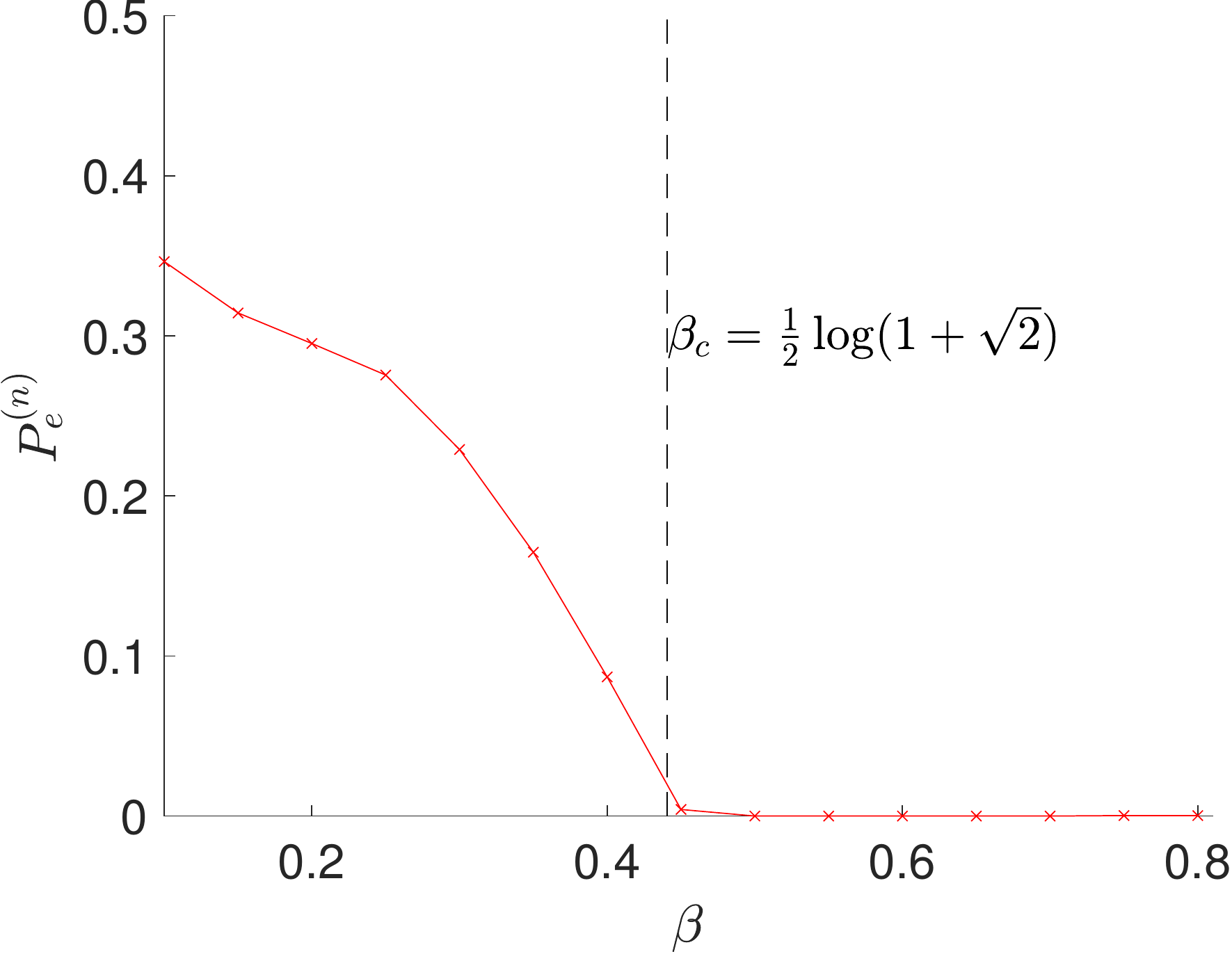}
  \vspace*{-0.3cm}
  \caption{Lattice graph: $P^{(n)}_e$ versus $\beta$ for $n = 2601$, $h=0$, $S=0$, $p=0.3$.}
  \label{fig:error_beta_lattice}
  \vspace*{-0.6cm}
\end{figure}

\section{Conclusion}
\label{sec:conclusion}
In this paper, we analyzed the asymptotic accuracy of supermajority sentiment detection in social networks with an external influence. We related the detection accuracy to the asymptotic distribution of the average member sentiments in the network. We showed that in several graphs such as Empty graph, Chain graph, Complete graph, Lattice graph, when the average member sentiment stays away from the supermajority threshold level, the detection is asymptotically accurate; otherwise, the detection is  inaccurate. 

\bibliographystyle{unsrt}
\bibliography{reflist}

\end{document}